\documentclass[a4paper,11pt]{amsart}
\usepackage[utf8]{inputenc}
\usepackage{amsaddr}
\usepackage{calrsfs}
\usepackage{graphicx,color}
\usepackage[english]{babel}
\usepackage{graphicx}
\usepackage{amsmath}
\usepackage{epstopdf}
\usepackage{float}
\usepackage{color}

\usepackage{amssymb}

\numberwithin{equation}{section}

\theoremstyle{definition}

\theoremstyle{plain}
\newtheorem{thm}{Theorem}[section]

\newtheorem{cor}{Corollary}[section]

\theoremstyle{definition}
\newtheorem{rem}{Remark}[section]

\newcommand{\R}{\mathbb{R}}
\newcommand{\E}{\mathbb{E}}
\renewcommand{\P}{\mathbb{P}}

\newcommand{\F}{\mathcal{F}}

\begin{document}
\title[geometric Asian power options]
{The evaluation of geometric Asian power options under time changed mixed fractional Brownian motion }

\date{\today}

\author[Shokrollahi]{Foad Shokrollahi}
\address{Department of Mathematics and Statistics, University of Vaasa, P.O. Box 700, FIN-65101 Vaasa, FINLAND}
\email{foad.shokrollahi@uva.fi}

\begin{abstract}
The aim of this paper is to evaluate geometric Asian option by a mixed fractional subdiffusive Black-Scholes model. We derive a pricing formula for geometric Asian option when the underlying stock follows a time changed mixed fractional Brownian motion. We then apply the results to price Asian power options on the stocks that pay constant dividends when the payoff is a power function. Finally, lower bound of Asian options and some special cases are provided.
\end{abstract}





\maketitle

\section{Introduction}\label{sec:1}
A standard option (also called plain vanilla) is a financial contract which gives the owner of the contract the right, but not the obligation, to buy or sell a specified asset to a prespecified price (strike price) at a prespecified time (maturity). The specified asset (underlying asset) can be for example stocks, indexes, currencies, bonds or commodities. The option can be either a call option, which gives the owner the right to buy the underlying asset, or it can be a put option, which gives the owner the right to sell the underlying asset. Moreover, the option can either only be exercised at maturity, European option, or it can be exercised at any time before maturity, American option. Path dependent options are options whose payoffs are affected by how the price of the underlying stock at maturity was reached, and the price path of the underlying stock. One particular path dependent option, called Asian option, will be of main focus throughout this research. The average price of the underlying asset can either determine the underlying settlement price (average price Asian options) or the option strike price (average strike Asian options). Furthermore, the average prices can be calculated using either the arithmetic mean or the geometric mean. The type of Asian option that will be examined throughout this research is geometric Asian option.

Over the past three decades, academic researchers and market practitioners have
developed and adopted different models and techniques for option valuation. The
path-breaking work on option pricing was undertaken by Black and Scholes $(BS)$ \cite{black1973pricing} in 1973. In the $BS$ model has been assumed
that the asset price dynamics are governed by a geometric Brownian motion. However, in the last few years based on some empirical studies, it has been shown that the geometric Brownian motion model cannot capture many
of the characteristic features of prices, such as: heavy tailed, long-range correlations, lack of scale invariance, periods of
constant values, and etc. Fractional Brownian motion has been suggested to display the long-range dependence and fluctuation observed in the empirical
data \cite{wang2012pricing,zhang2009equity,necula2002option}. Since fractional Brownian motion is neither a Markov process nor a semi-martingale, then we cannot
use the usual stochastic calculus to analyze it. Further, fractional Brownian motion admits arbitrage in a
complete and frictionless market. To get around this problem and to take into account the long memory property, it has
been proposed that it is reasonable to use the mixed  fractional Brownian motion $(mfBm)$ to capture the
fluctuations of financial asset \cite{el2003fractional,mishura2008stochastic,cheridito2001mixed}.

The $mfBm$ is a linear combination of the Brownian motion and fractional Brownian motion with Hurst index $H\in(\frac{1}{2}, 1)$, defined on the filtered probability $(\Omega, \F, \P)$ for any $t\in \R^+$ by:

\begin{eqnarray}
M_t^H(a, b)=aB(t)+bB^H(t),
\label{eq:1}
\end{eqnarray}

where $B(t)$ is a Brownian motion, and  $B^H(t)$ is a independent fractional Brownian motion with Hurst index $H$. Cheridito \cite{cheridito2001mixed} proved that, for $H\in(\frac{3}{4}, 1)$, the mixed model is equivalent to the Brownian motion and hence it is also arbitrage free. For $H\in(\frac{1}{2}, 1)$, Mishura and Valkeila \cite{mishura2002absence} demonstrated that the mixed model is arbitrage free. Rao \cite{rao2016pricing} discussed geometric Asian power option under $mfBm$. To see more about the mixed model, one can refer to Refs \cite{mishura2008stochastic,cheridito2001mixed,shokrollahi2016pricing}.

In order to describe properly financial data exhibiting periods of constant values, Magdziarz  \cite{magdziarz2012anomalous} introduced subdiffusive
geometric Brownian motion

 \begin{eqnarray}
X_\alpha(t)=X(T_\alpha(t)),
\label{eq:2}
\end{eqnarray}
where $X(t)$ is a geometric Brownian motion, $T_\alpha(t)$ is the inverse $\alpha$-stable subordinator with parameter $\alpha\in (0, 1)$. Magdziarz
pointed out that this model is arbitrage free but incomplete, and based on the subdiffusive
geometric Brownian motion obtained the corresponding subdiffusive $BS$ formula for the fair price of European options. Within the framework of subdiffusive theory, numerous scholars continue to investigate financial problems identified considered  in Magdziarz’s pioneer work of subdiffusion finance in 2009. These include the pricing formulas of European option and European currency option under subdiffusive fractional $BS$ and subdiffusive mixed fractional $BS$ models \cite{guo2014pricing,shokrollahi2016pricing,gu2012time}.

In this research, inspired by the works \cite{guo2014pricing} and \cite{shokrollahi2016pricing}, we introduce a pricing formula for geometric Asian options under time changed mixed fractional $BS$ model.
 We then apply the result to price geometric Asian power options that pay constant dividends when the payoff is a power function. We also provide some  special cases and lower bound for the Asian option price.
The rest of the paper is organized as follows. In Section \ref{sec:2}, some useful concepts and theorems of time changed mixed fractional process are introduced. In Section \ref{sec:3}, a brief introduction of Asian options is given. Analytical valuation formula for geometric Asian options is derived in Section \ref{sec:4}  and then applied to geometric Asian power options in Section \ref{sec:5}. The lower bound on the price of the Asian option is proposed in Section \ref{sec:6}.

\section{Auxiliary facts}\label{sec:2}
In this section, we recall some definitions and results about mixed fractional time changed process. More information about mixed fractional process can be found in \cite{guo2014pricing,shokrollahi2016pricing}.

The time-changed process $T_\alpha(t)$ is the inverse $\alpha$-stable subordinator defined as below
\begin{eqnarray*}
T_\alpha(t)=\inf\{\tau>0, U_\alpha(t)\geq t\}.
\label{eq:5}
\end{eqnarray*}

here $U_\alpha(\tau)_{\tau\geq 0}$ is a strictly increasing $\alpha$-stable Lévy process \cite{sato1999levy} with Laplace transform: $\E(e^{-uU_\alpha(\tau)})=e^{-\tau u^\alpha}$, $\alpha\in(0, 1)$.

 $U_\alpha(t)$ is $\frac{1}{\alpha}$ self-similar and $T_\alpha(t)$ is $\alpha$ self-similar, that is, for every $h>0$, $U_\alpha(ht)\triangleq h^{\frac{1}{\alpha}}U_\alpha(t)$  $T_\alpha(ht)\triangleq h^{\alpha}T_\alpha(t)$, here $\triangleq$ indicates that the random variables on both sides have the same distribution. Specially, when $\alpha\uparrow 1$, $T_\alpha(t)$ reduces to the physical time $t$. You can find more details about subordinator and its inverse processes in \cite{janicki1993simulation,piryatinska2005models}.

Consider the subdiffusion process
\begin{eqnarray*}
M_{\alpha}^ H(t)(a,b)=aW_{\alpha}(t)+bW_{\alpha}^ H(t)=aB(T_\alpha(t))+bB^H(T_\alpha(t)),
\label{eq:6}
\end{eqnarray*}
where $B(\tau)$ is a Brownian motion, $B^H(\tau)$ is a fractional Brownian motion with Hurst index $H$ and $T_\alpha(t)$ is inverse $\alpha$-subordinator which are supposed to be independent. When $a=0, b=1$, the results represented in \cite{gu2012time} and if $b=0, a=1$, then it is the process considered in \cite{magdziarz2009black}. In this research, we assume that $H\in(\frac{3}{4}, 1)$ and $(a, b)=(1, 1)$.

\begin{rem}
When $\alpha\uparrow 1$, the processes $W_{\alpha}(t)$ and $W_{\alpha}^ H(t)$ degenerate to $B(t)$ and $B^H(t)$, respectively. Then, $M_{\alpha}^ H(t)(a,b)$ reduces to the $mfBm$ in Eq. (\ref{eq:1}).
\end{rem}

\begin{rem}
From \cite{gu2012time,magdziarz2009black}, we know that $\E(T_\alpha(t))=\frac{t^\alpha}{\Gamma(\alpha+1)}$. Then, by applying $\alpha$-self-similar and non-decreasing
sample path of $T_\alpha(t)$, we have
\begin{eqnarray}
\E[(B(T_\alpha(t)))^2]&=&\frac{t^\alpha}{\Gamma(\alpha+1)}\\
\E[(B^H(T_\alpha(t)))^2]&=&\left(\frac{t^\alpha}{\Gamma(\alpha+1)}\right)^{2H}.
\label{eq:7}
\end{eqnarray}

\end{rem}

\section{Asian options}\label{sec:3}
The payoff of an Asian option is based on the difference between an asset's average price
over a given time period, and a fixed price called the strike price. Asian options are
popular because they tend to have lower volatility than options whose payoffs are based
purely on a single price point. It is also harder for big traders to manipulate an average
price over an extended period than a single price, so Asian option offers further
protection against risk. The Asian call and put options have a payoff that is calculated with an average value of
the underlying asset over a specific period. The payoff for an Asian call and put option with strike price $K$ and expiration time $T$ is $(\bar{S}(T)-K)_+$ and $(K-\bar{S}(T))_+$ respectively, where $\bar{S}(T)$ is the average price of the underlying asset over the prespecified interval.
Since Asian options are less expensive than their European counterparts, they are
attractive to many different investors.
Apart from the regular Asian option there also exists Asian strike option. An Asian
strike call option guarantees the holder that the average price of an underlying asset is
not higher than the final price. The option will not be exercised if the average price of
the underlying asset is greater than the final price. The holder of an Asian strike put
option makes sure that the average price received for the underlying asset is not less
than what the final price will provide. The payoff for
an Asian strike call and put option is $(\bar{S}(T)-S(T))_+$ and $(S(T)-\bar{S}(T))_+$ respectively, where $S(T)$ is the value of underlying stock at maturity date $T$.

Asian options are divided into two different types, when calculating the average, the
geometric Asian option
\begin{eqnarray*}
G(T)=\exp\left\{\frac{1}{T}\int_0^T\ln S(t)dt\right\},
\label{eq:3}
\end{eqnarray*}

and the arithmetic Asian option.

\begin{eqnarray*}
A(T)=\frac{1}{T}\int_0^T S(t)dt.
\label{eq:4}
\end{eqnarray*}
We assume that the prespecified interval $[0, T]$ is fixed, then will price the geometric Asian option in the continuous average case under time changed mixed fractional Brownian motion environment.

\section{Pricing model of geometric Asian option}\label{sec:4}

In order to derive an Asian option pricing formula in a time changed mixed fractional market, we
make the following assumptions:
\begin{enumerate}
\item[(i)] the price of underlying stock at time $t$ is given by

\begin{eqnarray}
S_t&&=S_0\exp\Big\{(r-q)T_\alpha(t)+\sigma W_{\alpha}(t)+\sigma W_{\alpha}^ H(t)\nonumber\\
&&-\frac{1}{2}\sigma^2\frac{t^\alpha}{\Gamma(\alpha+1)}-\frac{1}{2}\sigma^2\left(\frac{t^\alpha}{\Gamma(\alpha+1)}\right)^{2H}\Big\},\quad 0<t<T,
\label{eq:8}
\end{eqnarray}
where $H\in(\frac{3}{4}, 1)$, $\alpha\in (\frac{1}{2}, 1)$ and $\alpha H>1$.
\item[(ii)] there are no transaction costs in buying or selling the stocks or option.
\item[(iii)] the risk free interest rate $r$  and dividend rate $q$  are known and constant through time.
\item[(iv)] the option can be exercised only at the maturity time.
\end{enumerate}

From Eq. (\ref{eq:8}), we know that $\ln S_t\simeq N(u, v)$, where

\begin{eqnarray}
u&=&\ln S(0)+(r-q)T_\alpha(t)-\frac{1}{2}\sigma^2\frac{t^\alpha}{\Gamma(\alpha+1)}-\frac{1}{2}\sigma^2\left(\frac{t^\alpha}{\Gamma(\alpha+1)}\right)^{2H}\\
v&=&\sigma^2\frac{t^\alpha}{\Gamma(\alpha+1)}+\sigma^2\left(\frac{t^\alpha}{\Gamma(\alpha+1)}\right)^{2H}.
\label{eq:9}
\end{eqnarray}

Let $C(S(0), T)$ be the price of a European
call option at time $0$ with strike price $K$ and that matures at time $T$. Then,  from \cite{guo2014pricing}, we can get

\begin{eqnarray*}
C(S(0), T)=S(0)e^{-qT}\phi(d_1)-Ke^{-rT}\phi(d_2),
\label{eq:10}
\end{eqnarray*}

where
\begin{eqnarray*}
d_1&=&\frac{\ln\frac{S_0}{K}+(r-q+\frac{\hat{\sigma}^2}{2})T}{\hat{\sigma}\sqrt{T}},\quad d_2=d_1-\hat{\sigma}\sqrt{T},\\
\hat{\sigma}^2&=&\sigma^2\frac{T^{\alpha-1}}{\Gamma(\alpha)}+\sigma^2\left(\frac{T^{\alpha-1}}{\Gamma(\alpha)}\right)^{2H},
\label{eq:11}
\end{eqnarray*}
and $\phi(.)$ denotes cumulative normal density function.

Under the above assumptions (i)-(iv), we obtain the value of the geometric Asian call option by the following theorem

\begin{thm}

Suppose the stock price $S_t$ satisfied Eq. (\ref{eq:8}). Then, under the risk-neutral probability measure, the value of geometric Asian call option $C(S(0), T)$ with strike price $K$ and maturity time $T$ is given by

\begin{eqnarray}
C(S(0), T)&=&S(0)\exp\Bigg\{-rT+(r-q)\frac{T^{\alpha}}{\Gamma(\alpha+2)}+\frac{\sigma^2(-T)^{\alpha}}{2\Gamma(\alpha+3)}\nonumber\\
&&-\frac{\sigma^2T^{2\alpha H}}{4(2\alpha H+1)(\alpha H+1)(\Gamma(\alpha+1))^{2H}}\Bigg\}\phi(d_1)-Ke^{-qT}\phi(d_2),
\label{eq:12}
\end{eqnarray}
where
\begin{eqnarray*}
d_2&=&\frac{\mu_G-\ln K}{\sigma_G},\quad d_1=d_2+\sigma_G,\\
\mu_G&=&\ln S(0)+(r-q-\frac{\sigma^2}{2})\frac{T^{\alpha}}{\Gamma(\alpha+2)}-\frac{\sigma^2T^{2\alpha H}}{2(2\alpha H+1)(\Gamma(\alpha+1))^{2H}},\\
\sigma_G^2&=&\frac{\sigma^2T^{\alpha}}{\Gamma(\alpha+2)}+\frac{\sigma^2(-T)^{\alpha}}{\Gamma(\alpha+3)}+\frac{\sigma^2T^{2\alpha H}}{(2\alpha H+2)(\Gamma(\alpha+1))^{2H}},
\label{eq:13}
\end{eqnarray*}
the interest rate $r$ and the dividend rate $q$ are constant over time and $\phi(.)$ denotes cumulative normal density function.
\label{thm:1}
\end{thm}

\begin{proof}
Suppose
\begin{eqnarray*}
L(T)=\frac{1}{T}\int_0^T\ln S(t)dt.
\label{eq:14}
\end{eqnarray*}
Then
\begin{eqnarray}
G(T)=e^{L(T)}.
\label{eq:15}
\end{eqnarray}
We know that $\ln S_t\simeq N(u, v)$, then it is clear that the random variable $L(T)$ has Gaussian distribution under the risk-neutral
probability measure. We will now compute its mean and variance under the risk-neutral probability measure. Let $\E$ denote
the expectation and, $\mu_G$ and $\sigma_G^2$ denote the mean and the variance of the random variable $\E$ under the risk-neutral
probability measure. Note that

\begin{eqnarray*}
\mu_G&=&\E[L(T)]=\frac{1}{T}\int_0^T\E[\ln S(t)]dt\nonumber\\
&=&\ln S(0)+\frac{1}{T}\int_0^T(r-q)\frac{t^{\alpha}}{\Gamma(\alpha+1)}dt-\frac{\sigma^2}{2T}\int_0^T\left[\frac{t^{\alpha}}{\Gamma(\alpha+1)}+\frac{t^{2\alpha H}}{(\Gamma(\alpha+1))^{2H}}\right]dt\nonumber\\
&=&\ln S(0)+(r-q)\frac{T^{\alpha}}{\Gamma(\alpha+2)}-\frac{\sigma^2T^\alpha}{2\Gamma(\alpha+2)}-\frac{\sigma^2T^{2\alpha H}}{(4\alpha H+2)(\Gamma(\alpha+1))^{2H}},
\label{eq:16}
\end{eqnarray*}
and
\begin{eqnarray*}
\sigma_G^2&=&Var[L(T)]=\E[(L(T)-\mu_G)^2]\nonumber\\
&=&\frac{\sigma^2}{T^2}\int_0^T\int_0^T\left(\E[W_\alpha(t)W_\alpha(\tau)]+\E[W_\alpha^H(t)W_\alpha^H(\tau)]\right)dtd\tau,
\label{eq:17}
\end{eqnarray*}

by independence of the processes $B(t), B^H(t)$ and $T_\alpha(t)$, we obtain

\begin{eqnarray*}
&=&\frac{\sigma^2}{T^2}\int_0^T\int_0^T\left(|\frac{t^{\alpha}}{\Gamma(\alpha+1)}|+|\frac{\tau^{\alpha}}{\Gamma(\alpha+1)}|-|\frac{(t-\tau)^{\alpha}}{\Gamma(\alpha+1)}|\right)dtd\tau\\
&+&\frac{\sigma^2}{T^2}\int_0^T\int_0^T\left(|\frac{t^{\alpha}}{\Gamma(\alpha+1)}|^{2H}+|\frac{\tau^{\alpha}}{\Gamma(\alpha+1)}|^{2H}-|\frac{(t-\tau)^{\alpha}}{\Gamma(\alpha+1)}|^{2H}\right)dtd\tau\\
&=&\frac{\sigma^2T^{\alpha}}{\Gamma(\alpha+2)}+\frac{\sigma^2(-T)^{\alpha}}{\Gamma(\alpha+3)}+\frac{\sigma^2T^{2\alpha H}}{(2\alpha H+2)(\Gamma(\alpha+1))^{2H}}.
\label{eq:18}
\end{eqnarray*}

From (\ref{eq:15}), we know that the random variable $G(T )$ is log-normally distributed, then $\ln G(T )\simeq N(\mu_G, \sigma_G^2)$. Let $I=\{x:e^x>K\}$ and $\phi(.)$
be the probability density function of a standard normal distribution, then the price of geometric Asian call option is given by the following computations

\begin{eqnarray*}
C(S(0), T)&=& e^{-rT}\E[(G(T)-K)^+]\\
&=&e^{-rT}\int_I(e^x-K)\frac{1}{\sqrt{2\pi}\sigma_G}\exp\left\{-\frac{(x-\mu_G)^2}{2\sigma_G^2}\right\}dx\\
&=&e^{-rT}\int_I(e^{\mu_G+z\sigma_G }-K)\frac{1}{\sqrt{2\pi}\sigma_G}\exp\left\{-\frac{(x-\mu_G)^2}{2\sigma_G^2}\right\}\varphi(z)dz\\
&=&e^{-rT+\mu_G+\frac{1}{2}\sigma_G^2}\int_{-d_2}^{\infty}e^{-\frac{1}{2}(z-\sigma_G)^2}dz-Ke^{-rT}\int_{-d_2}^{-\infty}\varphi(z)dz\\
&=&e^{-rT+\mu_G+\frac{1}{2}\sigma_G^2}\int_{-d_2-\sigma_G}^{\infty}\varphi(z)dz-Ke^{-rT}\int_{-\infty}^{d_2}\varphi(z)dz\\
&=&e^{-rT+\mu_G+\frac{1}{2}\sigma_G^2}\int_{-\infty}^{d_2+\sigma_G}\varphi(z)dz-Ke^{-rT}\int_{-\infty}^{d_2}\varphi(z)dz\\
&=&e^{-rT+\mu_G+\frac{1}{2}\sigma_G^2}\phi(d_1)-Ke^{-rT}\phi(d_2),\\
&=&S(0)\exp\Bigg\{-rT+(r-q)\frac{T^{\alpha}}{\Gamma(\alpha+2)}+\sigma^2\frac{(-T)^{\alpha}}{2\Gamma(\alpha+3)}\nonumber\\
&&-\sigma^2\frac{T^{2\alpha H}}{4(2\alpha H+1)(\alpha H+1)(\Gamma(\alpha+1))^{2H}}\Bigg\}\phi(d_1)-Ke^{-qT}\phi(d_2),
\label{eq:19}
\end{eqnarray*}

here

\begin{eqnarray*}
I&=&\{x:e^x>K\}=\{z:e^{\mu_G+z\sigma_G}>K\}\\
&=&\{z:\mu_G+z\sigma_G>\ln K\}=\{z:z>-d_2\},
\label{eq:20}
\end{eqnarray*}
thus we obtain the pricing formula.
\end{proof}
Moreover, using the put–call parity, the valuation model for a geometric Asian put option under time changed mixed fractional $BS$ model can be written

\begin{eqnarray}
P(S(0), T)&=&Ke^{-qT}\phi(-d_2)-S(0)\exp\Bigg\{-rT+(r-q)\frac{T^{\alpha}}{\Gamma(\alpha+2)}+\frac{\sigma^2(-T)^{\alpha}}{2\Gamma(\alpha+3)}\nonumber\\
&&-\frac{\sigma^2T^{2\alpha H}}{4(2\alpha H+1)(\alpha H+1)(\Gamma(\alpha+1))^{2H}}\Bigg\}\phi(-d_1),
\label{eq:21}
\end{eqnarray}

where $d_1$ and $d_2$ are defined previously.

Letting $\alpha\uparrow 1$, then the stock price follows the $mfBm$ shown below

\begin{eqnarray}
S_t&&=S_0\exp\Big\{(r-q)T+\sigma B(t)+\sigma B^ H(t)\nonumber\\
&&-\frac{1}{2}\sigma^2t-\frac{1}{2}\sigma^2t^{2H}\Big\},\quad 0<t<T,
\label{eq:22}
\end{eqnarray}
and the result is presented below.
\begin{cor}
The value of geometric Asian call option with maturity $T$ and strike $K$, whose stock price follows Eq. (\ref{eq:22}), is given by

\begin{eqnarray}
&&C(S(0), T)=\nonumber\\
&&S(0)\exp\Bigg\{-\frac{1}{2}(r+q)T-\frac{\sigma^2T}{12}-\frac{\sigma^2T^{2H}}{4(2H+1)(H+1)}\Bigg\}\phi(d_1)-Ke^{-qT}\phi(d_2),
\label{eq:23}
\end{eqnarray}
where
\begin{eqnarray*}
d_2&=&\frac{\mu_G-\ln K}{\sigma_G},\quad d_1=d_2+\sigma_G,\\
\mu_G&=&\ln S(0)+\frac{1}{2}(r-q-\frac{\sigma^2}{2})T-\frac{\sigma^2T^{2H}}{2(2 H+1)},\\
\sigma_G^2&=&\frac{\sigma^2T}{3}+\frac{\sigma^2T^{2H}}{(2H+2)},
\label{eq:24}
\end{eqnarray*}
which is consistent with result in \cite{rao2016pricing}.
\end{cor}

\section{Pricing model of Asian power option}\label{sec:5}

In this section, we consider the pricing model of Asian power call option with strike price $K$ and maturity time $T$ under time changed mixed fractional $BS$ model where the payoff function is $(G^n(T)-K)^+$ for some constant integer $n\geq 1$.

\begin{thm}

Suppose the stock price $S_t$ satisfied Eq. (\ref{eq:8}). Then, under the risk-neutral probability measure the value of geometric Asian power call option $C(S(0), T)$ with strike price $K$, maturity time $T$ and payoff function $(G^n(T)-K)^+$ is given by

\begin{eqnarray}
C(S(0), T)&=&S(0)\exp\Bigg\{-rT+(r-q)\frac{nT^{\alpha}}{\Gamma(\alpha+2)}-\frac{(n-n^2)\sigma^2T^\alpha}{2\Gamma(\alpha+2)}+\frac{n^2\sigma^2(-T)^{\alpha}}{2\Gamma(\alpha+3)}\nonumber\\
&&-\frac{n\sigma^2T^{2\alpha H}}{(4\alpha H+2)(\Gamma(\alpha+1))^{2H}}-\frac{n^2\sigma^2T^{2\alpha H}}{(4\alpha H+4)(\Gamma(\alpha+1))^{2H}}\Bigg\}\phi(f_1)\nonumber\\
&-&Ke^{-qT}\phi(f_2),
\label{eq:25}
\end{eqnarray}
where
\begin{eqnarray*}
f_2&=&\frac{\mu_G-\frac{1}{n}\ln K}{\sigma_G},\quad f_1=f_2+n\sigma_G,\\
\mu_G&=&\ln S(0)+(r-q-\frac{\sigma^2}{2})\frac{T^{\alpha}}{\Gamma(\alpha+2)}-\frac{\sigma^2T^{2\alpha H}}{2(2\alpha H+1)(\Gamma(\alpha+1))^{2H}},\\
\sigma_G^2&=&\frac{\sigma^2T^{\alpha}}{\Gamma(\alpha+2)}+\frac{\sigma^2(-T)^{\alpha}}{\Gamma(\alpha+3)}+\frac{\sigma^2T^{2\alpha H}}{(2\alpha H+2)(\Gamma(\alpha+1))^{2H}},
\label{eq:26}
\end{eqnarray*}
the interest rate $r$ and the dividend rate $q$ are constant over time and $\varphi(.)$ denotes cumulative normal density function.
\label{thm:2}
\end{thm}

\begin{proof}
The payoff function for Asian power option is $(G^n(T)-K)^+=(e^{nL(T)}-K)^+$, then applying similar computation in Theorem \ref{thm:1}, we obtain

\begin{eqnarray*}
C(S(0), T)&=& e^{-rT}\E[(G^n(T)-K)^+]\\
&=&e^{-rT}\int_I(e^{nx}-K)\frac{1}{\sqrt{2\pi}\sigma_G}\exp\left\{-\frac{(x-\mu_G)^2}{2\sigma_G^2}\right\}dx\\
&=&e^{-rT}\int_I(e^{n(\mu_G+z\sigma_G) }-K)\frac{1}{\sqrt{2\pi}\sigma_G}\exp\left\{-\frac{(x-\mu_G)^2}{2\sigma_G^2}\right\}\varphi(z)dz\\
&=&e^{-rT+n\mu_G+\frac{1}{2}n^2\sigma_G^2}\int_{-f_2}^{\infty}e^{-\frac{1}{2}(z-n\sigma_G)^2}dz-Ke^{-rT}\int_{-f_2}^{-\infty}\varphi(z)dz\\
&=&e^{-rT+n\mu_G+\frac{1}{2}n^2\sigma_G^2}\int_{-f_2-n\sigma_G}^{\infty}\varphi(z)dz-Ke^{-rT}\int_{-\infty}^{f_2}\varphi(z)dz\\
&=&e^{-rT+n\mu_G+\frac{1}{2}n^2\sigma_G^2}\int_{-\infty}^{f_2+n\sigma_G}\varphi(z)dz-Ke^{-rT}\int_{-\infty}^{f_2}\varphi(z)dz\\
&=&e^{-rT+n\mu_G+\frac{1}{2}n^2\sigma_G^2}\phi(f_1)-Ke^{-rT}\phi(f_2),\\
&=&S(0)\exp\Bigg\{-rT+(r-q)\frac{nT^{\alpha}}{\Gamma(\alpha+2)}-\frac{(n-n^2)\sigma^2T^\alpha}{2\Gamma(\alpha+2)}+\frac{n^2\sigma^2(-T)^{\alpha}}{2\Gamma(\alpha+3)}\nonumber\\
&&-\frac{n\sigma^2T^{2\alpha H}}{(4\alpha H+2)(\Gamma(\alpha+1))^{2H}}-\frac{n^2\sigma^2T^{2\alpha H}}{(4\alpha H+4)(\Gamma(\alpha+1))^{2H}}\Bigg\}\phi(f_1)\nonumber\\
&-&Ke^{-qT}\phi(f_2),
\label{eq:27}
\end{eqnarray*}

here

\begin{eqnarray*}
I&=&\{x:e^{nx}>K\}=\{z:e^{n(\mu_G+z\sigma_G)}>K\}\\
&=&\{z:\mu_G+z\sigma_G>\frac{1}{n}\ln K\}=\{z:z>-f_2\},
\label{eq:28}
\end{eqnarray*}
thus the proof is completed.
\end{proof}

\section{Lower bound of the Asian option price}\label{sec:6}

The aim of this section is to obtain the lower bound on the price of the Asian option. The next theorem shows that the normal distribution is stable when the random variables are jointly normal.

\begin{thm} (\cite{hoffman1994probability}) The conditional distribution of $\ln S_{t_i}$ given $\ln G(T)$ is a normal distribution
\begin{eqnarray*}
(\ln S_{t_i}|\ln G(T)=z)\simeq N(\mu_i+(z-\mu_G)\frac{\lambda_i}{\sigma_G^2}, \sigma_i^2-\frac{\lambda_i^2}{\sigma_G^2}),\quad i=1,...,n,
\label{eq:29}
\end{eqnarray*}
where
\begin{eqnarray*}
\mu_i&=&\ln S(0)+(r-q)T_\alpha(t_i)-\frac{1}{2}\sigma^2\frac{t_i^\alpha}{\Gamma(\alpha+1)}-\frac{1}{2}\sigma^2\left(\frac{t_i^\alpha}{\Gamma(\alpha+1)}\right)^{2H}\\
\sigma_i^2&=&\sigma^2\frac{t_i^\alpha}{\Gamma(\alpha+1)}+\sigma^2\left(\frac{t_i^\alpha}{\Gamma(\alpha+1)}\right)^{2H},
\label{eq:30}
\end{eqnarray*}
$\lambda_i=Cov(\ln S_{t_i}, \ln G(T))$, $0\leq t_1<t_2<...<t_n\leq T$, $T_\alpha(t)$ is inverse $\alpha$-stable subordinator and, $\mu_G$ and $\sigma_G^2$ are defined in Theorem \ref{thm:1}.

Moreover, $(S_{t_i}|\ln G(T))$ has a lognormal distribution and

\begin{eqnarray}
&&\E\left[S_{t_i}|\ln G(T)=z\right]\nonumber\\
&&=\exp\left\{\mu_i+(z-\mu_G)\frac{\lambda_i}{\sigma_G^2}+\frac{1}{2} (\sigma_i^2-\frac{\lambda_i^2}{\sigma_G^2})\right\}\quad i=1,...,n.
\label{eq:31}
\end{eqnarray}

\label{thm:3}
\end{thm}

Now, we condition on the geometric average $G(T)$ in the pricing expresion of the Asian option

\begin{eqnarray*}
C(S(0), T)&=&e^{-rT}\E[(A(T)-K)^+]=e^{-rT}\E[\E[(A(T)-K)^+|G(T)]\\
&=&e^{-rT}\int_0^\infty\E[(A(T)-K)^+|G(T)=z]g(z)dz,
\label{eq:32}
\end{eqnarray*}
 where $g$ is the lognormal density function of $G$. Let

\begin{eqnarray*}
C_1&=&\int_0^K\E[(A(T)-K)^+|G(T)=z]g(z)dz,\\
C_2&=&\int_K^\infty\E[(A(T)-K)^+|G(T)=z]g(z)dz,
\label{eq:33}
\end{eqnarray*}
then $C(S(0), T)=e^{-rT}(C_1+C_2)$.
Sine the geometric average is less than arithmetic average $A(T)\geq G(T)$,
\begin{eqnarray}
C_2&=&\int_K^\infty\E[A(T)-K|G(T)=z]g(z)dz,
\label{eq:34}
\end{eqnarray}

from Theorem \ref{thm:3}, we can calculate $C_2$. Applying Jensen's inequality we obtain a lower bound on $C_1$

\begin{eqnarray}
C_1&=&\int_0^K\E[(A(T)-K)^+|G(T)=z]g(z)dz\nonumber\\
&\geq &\int_0^K\left(E[A(T)-K|G(T)=z]\right)^+g(z)dz\nonumber\\
&=&\int_{\tilde{K}}^K\E[A(T)-K|G(T)=z]g(z)dz=\tilde{C}_1.
\label{eq:35}
\end{eqnarray}
where $\tilde{K}=\left\{z|\E[A(T)|G(T)=z]=K\right\}$.

Eq. (\ref{eq:31}) enables us to obtain $\tilde{K}$, then we calculate the following expectation

\begin{eqnarray*}
\E[A(T)|G(T)=z]&=&\E\left[\frac{1}{n}\sum_{i=1}^nS_{t_i}|G(T)=z\right]=\frac{1}{n}\sum_{i=1}^n\E\left[S_{t_i}|G(T)=z\right]\\
&&=\frac{1}{n}\sum_{i=1}^n\exp\left(\mu_i+(\log z-\mu_G)\frac{\lambda_i}{\sigma_G^2}+\frac{1}{2} (\sigma_i^2-\frac{\lambda_i^2}{\sigma_G^2})\right).
\label{eq:36}
\end{eqnarray*}

\begin{thm}

A lower bound on the price of the Asian option with strike price $K$ and maturity time  $T$ is given by
\begin{eqnarray*}
\tilde{C}(S(0), T)&=&e^{-rT}(\tilde{C}_1+C_2)\\
&=&e^{-rT}\Big\{\frac{1}{n}\sum_{i=1}^n\exp(\mu_i+\frac{1}{2}\sigma_i^2)\phi\left(\frac{\mu_G-\ln \tilde{K}+\gamma_i}{\sigma_G}\right)\\
&&-K\phi\left(\frac{\mu_G-\ln \tilde{K}}{\sigma_G}\right)\Big\},
\label{eq:37}
\end{eqnarray*}
where all parameters are defined previously.
\label{thm:4}
\end{thm}

\begin{proof}
Collecting Eqs. (\ref{eq:34}) and (\ref{eq:35}) gives

\begin{eqnarray*}
\tilde{C}_1+C_2&=&\int_{\tilde{K}}^\infty\E[A(T)-K|G(T)=z]g(z)dz\\
&=&\int_{\tilde{K}}^\infty\E[A(T)|G(T)=z]g(z)dz-K\int_{\tilde{K}}^\infty g(z)dz\\
&=&\int_{\tilde{K}}^\infty\E\left[\frac{1}{n}\sum_{i=1}^nS_{t_i}|G(T)=z\right]g(z)dz-K\int_{\tilde{K}}^\infty g(z)dz\\
&=&\int_{\tilde{K}}^\infty\frac{1}{n}\sum_{i=1}^n\E\left[S_{t_i}|G(T)=z\right]g(z)dz-K\int_{\tilde{K}}^\infty g(z)dz\\
&=&\frac{1}{n}\sum_{i=1}^n\int_{\tilde{K}}^\infty\E\left[S_{t_i}|\ln G(T)=\ln z\right]g(z)dz-K\int_{\tilde{K}}^\infty g(z)dz.
\label{eq:38}
\end{eqnarray*}

From the proof of Theorem \ref{thm:1}, we obtain

\begin{eqnarray*}
K\int_{\tilde{K}}^\infty g(z)dz=K\phi\left(\frac{\mu_G-\ln \tilde{K}}{\sigma_G}\right),
\label{eq:39}
\end{eqnarray*}

and from Eq. (\ref{eq:31})
\begin{eqnarray*}
&&\int_{\tilde{K}}^\infty\E\left[S_{t_i}|\ln G(T)=\ln z\right]g(z)dz\\
&&=\int_{\tilde{K}}^\infty\exp\left(\mu_i+(\ln z-\mu_G)\frac{\lambda_i}{\sigma_G^2}+\frac{1}{2} (\sigma_i^2-\frac{\lambda_i^2}{\sigma_G^2})\right)g(z)dz\\
&&=\exp\left(\mu_i+\frac{1}{2} \sigma_i^2\right)\int_{\tilde{K}}^\infty\exp\left((\ln z-\mu_G)\frac{\lambda_i}{\sigma_G^2}-\frac{1}{2}\frac{\lambda_i^2}{\sigma_G^2}\right)g(z)dz.
\label{eq:40}
\end{eqnarray*}

Using the density of the lognormal distribution, we have

\begin{eqnarray*}
\int_{\tilde{K}}^\infty\frac{1}{\sqrt{2\pi}\sigma_G z}\exp\left((\ln z-\mu_G)\frac{\lambda_i}{\sigma_G^2}-\frac{1}{2}\frac{\lambda_i^2}{\sigma_G^2}-\frac{1}{2}(\frac{\mu_G-\ln z}{\sigma_G})^2\right)dz.
\label{eq:41}
\end{eqnarray*}

Making the change of variables $y=\frac{\mu_G-\ln z +\lambda_i}{\sigma_G}$ and $\frac{dy}{dz}=-\frac{1}{\sigma_G z}$, then we have

\begin{eqnarray*}
&&\int_{\frac{\mu_G-\ln z +\lambda_i}{\sigma_G}}^{-\infty}-\frac{1}{\sqrt{2\pi}}\exp\left((\frac{\lambda_i}{\sigma_G}-y)\frac{\lambda_i}{\sigma_G}-\frac{1}{2}\frac{\lambda_i^2}{\sigma_G^2}-\frac{1}{2}(y-\frac{\lambda_i}{\sigma_G})^2 \right)dy\\
&&=\int_{-\infty}^{\frac{\mu_G-\ln z +\lambda_i}{\sigma_G}}\frac{1}{\sqrt{2\pi}}\exp\left(-y\frac{\lambda_i}{\sigma_G}
+\frac{1}{2}\frac{\lambda_i^2}{\sigma_G^2}-\frac{1}{2}y^2-\frac{1}{2}\frac{\lambda_i^2}{\sigma_G^2}+y\frac{\lambda_i}{\sigma_G} \right)dy\\
&&=\int_{-\infty}^{\frac{\mu_G-\ln z +\lambda_i}{\sigma_G}}\frac{1}{\sqrt{2\pi}}\exp\left(-\frac{1}{2}y^2\right)dy=\phi\left(\frac{\mu_G-\ln \tilde{K}+\gamma_i}{\sigma_G}\right),
\label{eq:41}
\end{eqnarray*}

by collecting $\tilde{C}_1$ and $C_2$ the proof is completed.

\end{proof}

\bibliography{reference}
\bibliographystyle{ieeetr}

\end{document}